\documentclass{article}[11pt]
\usepackage[margin=1in]{geometry}

\usepackage{amssymb}
\usepackage{amsmath}
\usepackage{amsthm}
\usepackage{bbm}
\usepackage{xcolor}
\usepackage[normalem]{ulem}
\usepackage{hyperref}
\definecolor{ceruleanblue}{rgb}{0.16, 0.32, 0.75}
\definecolor{darkmidnightblue}{rgb}{0.0, 0.2, 0.4}
\definecolor{darkpastelgreen}{rgb}{0.01, 0.75, 0.24}
\definecolor{bleudefrance}{rgb}{0.19, 0.55, 0.91}
\hypersetup{
	colorlinks   = true,
	citecolor    = bleudefrance,
	linkcolor	  = bleudefrance,
	urlcolor     = bleudefrance
}
\usepackage[capitalize, nameinlink]{cleveref}
\usepackage{thmtools}
\usepackage{thm-restate}

\usepackage{verbatim}

\newcommand{\Var}{\operatorname{Var}}

\newcommand{\unif}{\operatorname{Unif}}

\DeclareMathOperator*{\E}{\mathbb{E}}

\usepackage{thmtools, thm-restate}
\newtheorem{theorem}{Theorem}

\newtheorem{claim}[theorem]{Claim}

\newtheorem{corollary}[theorem]{Corollary}

\definecolor{HighlightColor}{rgb}{0.70,0.05,0.70}
\definecolor{Del}{rgb}{0.7,0.05,0.05}

  \usepackage{nth}
  \usepackage{intcalc}

  \newcommand{\cAAAI}[1]{AAAI\ Conference\ on\ Artificial (AAAI)}

   \newcommand{\cESA}[1]{European\ Symposium\ on\ Algorithms\ (ESA)}

\title{Simple Analysis of Priority Sampling}
\author{Majid Daliri, Juliana Freire, Christopher Musco, Aécio Santos, Haoxiang Zhang\\\\
New York University, Tandon School of Engineering\\
\{daliri.majid, juliana.freire, cmusco, aecio.santos, haoxiang.zhang\}@nyu.edu \\
}
\date{}

\begin{document}
\maketitle
\begin{abstract}
    We prove a tight upper bound on the variance of the priority sampling method (aka sequential Poisson sampling). Our proof is significantly shorter and simpler than the original proof given by Mario Szegedy at STOC 2006, which resolved a conjecture by Alon, Duffield, Lund, and Thorup. 
\end{abstract}

% \newpage
% \setcounter{page}{1}
\section{Background}
Suppose we have a list of non-negative numbers $w_1, \ldots, w_n$. A common task in streaming and distributed algorithms is to collect a sample of this list, which can then be used to estimate arbitrary sums of a subset of the numbers. As toy examples, we might hope to estimate $\sum_{i: i \text{ is odd}} w_i$ or $\sum_{i: 100 \leq i \leq 200} w_i$. Importantly, the condition used to determine the subset will not be known in advance when collecting samples. 

It has been observed that to obtain accurate results for subset sum estimation, it is usually important to sample from $w_1, \ldots, w_n$ with ``probability proportional to size''. I.e., we want to collect a subset of size $k\ll n$ from this list in such a way that larger numbers are sampled with higher probability -- ideally proportional to or approximately proportional to their size. Such a subset will typically be more useful in estimating sums than a uniform sample. We will also refer to probability proportional to size as ``weighted sampling''. 

\subsection{Threshold Sampling}
% \as{Approximate in which sense?}
One simple approach for weighted sampling is the so-called \emph{Threshold Sampling} method, also referred to as Poisson sampling \cite{DuffieldLundThorup:2005}. Threshold sampling is used in computer science due to applications in \emph{sample coordination}, a topic beyond the scope of this note \cite{Flajolet:1990,CohenKaplan:2013}.
For each item $w_i$, we draw a uniform random variable $u_i \sim \unif[0,1]$. Then, we fix a threshold $\tau \geq 0$ and sample all numbers $w_i$ for 
which $\frac{u_i}{w_i} \leq \tau$. Evidently, $w_i$ gets sampled with probability:
\begin{align*}
p_i = \min(1, w_i\tau).
\end{align*}
The probabilities $p_1, \ldots, p_n$ are approximately proportional to the weights $w_1, \ldots, w_n$ (in fact, exactly proportional unless $w_i\tau > 1$ for some $i$).
The expected number of items sampled is upper bounded by $\sum_{i=1}^n p_i \leq \sum_{i=1}^n w_i\tau = \tau \cdot W$,
where $W = \sum_{i=1}^n w_i$. If $\tau = \frac{k}{W}$, the expected number of items sampled is $\leq k$. To use our samples to estimate the sum of a subset of items $\mathcal{I}\subseteq\{1, \ldots, n\}$, a natural approach is to apply the Horvitz-Thompson estimator:
\begin{align*}
\sum_{i\in \mathcal{I}} \hat{w}_i&\approx \sum_{i\in \mathcal{I}} {w}_i & &\text{where for all $i\in 1,\ldots, n$,} & \hat{w}_i&=\mathbbm{1}\left[\frac{u_i}{w_i}\leq \tau\right] \cdot \frac{w_i}{p_i}.
\end{align*}
Here $\mathbbm{1}[A]$ denotes the indicator random variable that evaluates to $1$ if the event $A$ is true and $0$ otherwise. 
It is not hard to see that $\E[\hat{w}_i] = w_i$, so $\sum_{i\in \mathcal{I}} \hat{w}_i$ is an unbiased estimate for the true subset sum. Since $\hat{w}_1, \ldots, \hat{w}_n$ are independent, the variance of this estimate is $\sum_{i\in\mathcal{I}} \Var[\hat{w}_i]$, a quantity that depends on the unknown set $\mathcal{I}$. So, in lieu of bounding variance, a common goal is to bound the \emph{total variance}, $\sum_{i=1}^n \Var[\hat{w}_i]$.\footnote{Other proxy performance measures besides total variance have also been studied, like the average variance for random subsets of a fixed size \cite{SzegedyThorup:2007}.}
% Again, since $\hat{w}_1, \ldots, \hat{w}_n$ are independent, this total variance equals $\Var[\hat{W}]$, where $\hat{W} = \sum_{i=1}^n \hat{w}_i$. 
To do so, first note that when $w_i\tau \geq 1$ (i.e., $p_i = 1$), we have $\Var[\hat{w}_i] = 0$. So, we can restrict our attention to terms for which $w_i\tau < 1$ (i.e., $p_i = w_i\tau)$. Specifically, letting $\mathcal{K}$ be a set containing all $i$ for which $w_i\tau < 1$, and setting $\tau = \frac{k}{W}$ (so we take at most $k$ samples in expectation), the total variance can bounded by:
\begin{align*}
% \Var[\hat{W}] 
% = 
\sum_{i=1}^n \Var[\hat{w}_i] 
= 
\sum_{i \in \mathcal{K}} \Var[\hat{w}_i] =
\sum_{i \in \mathcal{K}}  \frac{w_i^2}{p_i^2} \Var\left[\mathbbm{1}\left[\frac{u_i}{w_i}\leq \tau\right]\right] 
= \sum_{i \in \mathcal{K}}\frac{w_i^2}{p_i^2} p_i(1-p_i) 
\leq  \sum_{i \in \mathcal{K}} \frac{w_i^2}{p_i} 
= \sum_{i \in \mathcal{K}}  \frac{w_i}{\tau} 
\leq \frac{W}{\tau} = \frac{W^2}{k}.
\end{align*}
This upper bound of $\frac{W^2}{k}$ for threshold sampling is known to be optimal in the sense that any sampling scheme generating a sequence of random variables $\hat{w}_1, \ldots, \hat{w}_n$ such that $\E[\hat{w}_i] = w_i$ cannot have a lower total variance if the expected number of non-zero variables is $\leq k$ \cite{DuffieldLundThorup:2007}.

\subsection{Priority Sampling}
% \as{In this paper we define priority sampling using the $k^{th}$ smallest values of $\frac{u_i}{w_i}$. But the original paper defines it as the $k^{th}$ largest $\frac{w_i}{u_i}$. Given this paper is focused on Priority sampling, it would be good to have some discussion on the equivalence of these two different formulations. It was very confusing to me at first to see this $k^{th}$-minimum-based definition in Edith Cohen's papers without much explanation, and then seeing a different definition in the original DLT paper.
% }\chris{I added a footnote to address this below}\\

While variance optimal, a disadvantage of threshold sampling is that it only guarantees that $k$ samples are taken \emph{in expectation}. Ideally, we want a scheme that samples \emph{exactly} $k$ items, while still sampling with probabilities (approximately) proportional to the weights $w_1, \ldots, w_n$. Many such schemes exist, including pivotal sampling, reservoir sampling methods, and conditional Poisson sampling \cite{Tille:2023}. In computer science, one method of particular interest is \emph{Priority Sampling}, which was introduced to the field by \cite{DuffieldLundThorup:2004}, but also studied in statistics under the name ``Sequential Poisson Sampling'' \cite{Ohlsson:1998}. Similar to threshold sampling,  priority sampling is often preferred in computer science over methods like pivotal sampling due to applications in coordinated random sampling. 
% \as{preferred over what?  sampling methods that do not allow coordination?}

In fact, priority sampling is almost identical to threshold sampling. The one (major) difference is that the threshold $\tau$ is chosen adaptively to equal the $(k+1)^\text{st}$ smallest item in the list $\{\frac{u_1}{w_1}, \ldots, \frac{u_n}{w_n}\}$. Let $\mathcal{S}$ contain all values of $i$ such that $\frac{u_i}{w_i} < \tau$ (i.e., the indices of the $k$ smallest items in the list).\footnote{We note that some papers, including the early work in~\cite{DuffieldLundThorup:2004}, define as ``priorities'' $\frac{w_1}{u_1}, \ldots, \frac{w_n}{u_n}$ and select the $k$ indices with the \emph{largest} priority. This is of course equivalent to selecting the indices with the smallest values of $\frac{u_i}{w_i}$.
% \as{Is there anything we can say about why we prefer this definition (e.g. because it fits in the bottom-$k$ sampling framework, and our proof applies to those as well)?}\chris{I feel like I don't really want to advocate for one way over another. I do think we tend to think about the "minimum" version because we are used to e.g. KMV, but really neither way is better.}\as{I didn't mean prefer as in one is better than the other, but just explain why we choose this version. But I'm ok with not saying anything here too}
}
We define
\begin{align}
\label{eq:priority_sampling}
     \hat{w}_{i} = \begin{cases}\frac{w_i}{\min(1, w_i \tau )} & i\in S\\
     0 & i\notin S
     \end{cases}.
\end{align}
As before, to estimate the sum of a subset $\mathcal{I}\subseteq \{1, \ldots, n\}$, 
we return $\sum_{i\in \mathcal{I}}\hat{w}_i$. Analyzing this estimator is trickier than 
threshold sampling because $\tau$ is now a random number that depends on $u_1,\ldots, u_n$. As a result, $\hat{w}_1, \ldots, \hat{w}_n$ are no longer independent random variables. However, the following (surprising) fact was proven in \cite{DuffieldLundThorup:2007} (we include a complete proof in \Cref{sec:app} as well):
\begin{restatable}{fact}{expcov}
\label{fact:expcov}
    Let $\hat{w}_1, \ldots, \hat{w}_n$ be as defined in \eqref{eq:priority_sampling}. For all $i$, $\E[\hat{w}_i] = w_i$ and for all $i\neq j$, $\E[\hat{w}_i\hat{w}_j] = w_iw_j$. In other words, the random variables are equal to $w_1, \ldots, w_n$ in expectation, and are pairwise uncorrelated.
\end{restatable}
% \as{Why is it natural? I can see that variance is a natural performance metric but not the relation to pairwise correlation.}\chris{this is a bit subtle. I'm trying to clarify in the previous section and a bit hear. Let me know if it makes sense.} \as{Yes, it is clearer now!}\\

It follows that for any subset $\mathcal{I}$, $\mathbb{E}\left[\sum_{i\in \mathcal{I}} \hat{w}_i\right] = \sum_{i\in \mathcal{I}} {w}_i$. So, samples collected via priority sampling can be used to obtain an unbiased estimate for subset sums. Additionally, since $\hat{w}_1,\ldots, \hat{w}_n$ are pairwise uncorrelated, we have that $\Var[\sum_{i\in \mathcal{I}} \hat{w}_i] = \sum_{i\in \mathcal{I}} \Var[\hat{w}_i]$, as was the case for threshold sampling. So, a natural goal is still to bound the total variance $\sum_{i=1}^n \Var[\hat{w}_i]$. It was shown in \cite{AlonDuffieldLund:2005} that $\sum_{i=1}^n \Var[\hat{w}_i] = O\left(\frac{W^2}{k}\right)$, where $W = \sum_{i=1}^n w_i$. This matches the $\frac{W^2}{k}$ bound for threshold sampling up to a constant factor. However, it was conjectured in that work that the bound could be improved to $\frac{W^2}{k-1}$, which is only \emph{just worse} than the optimal $\frac{W^2}{k}$. This conjecture was resolved in a 2006 paper by Szegedy:
\begin{theorem}[\cite{Szegedy:2006}, Thm. 4]
\label{thm:mainvar}
 Let $\hat{w}_1, \ldots, \hat{w}_n$ be as defined in \eqref{eq:priority_sampling}, let $\hat{W} = \sum_{i=1}^n \hat{w}_i$, and let ${W} = \sum_{i=1}^n {w}_i$.
 \vspace{-1em}
 \begin{align*}
     \Var[\hat{W}] =\sum_{i=1}^n \Var[\hat{w}_i] \leq \frac{W^2}{k-1}.
 \end{align*}
\end{theorem}
We note that such a bound is also known to hold for other related sampling methods amenable to sample coordination, like the  successive weighted sampling without replacement (PPSWOR) method \cite{Cohen:2015}.

Szegedy's proof of \Cref{thm:mainvar} is quite involved, as it is based on an explicit integral formula for the total variance, and several pages of detailed calculations. We provide a simple alternative proof below. It is important to note that Szegedy's proof actually gives an ``instance optimality'' result, which can be stronger than \Cref{thm:mainvar} for some sets of weights. We discuss the difference in detail in \Cref{app:refinement_thrm}.

% \revised{We also present a straightforward proof \Cref{sec:app_near_opt} for Theorem 7 in \cite{Szegedy:2006}, which offers a tighter result compared to Theorem 2. This is used to demonstrate the near-optimality of the priority sampling.} We provide a simple alternative proof below. \chris{need to modify the sentence above to imply that he still proves a slightly stronger result, but not much}

\section{Main Analysis}

\label{sec:main}
As in prior work (e.g. \cite{DuffieldLundThorup:2007}) we introduce a new random variable $\tau_i$ for each item $i$. $\tau_i$ is equal to the $k^\text{th}$ smallest value of $\frac{u_j}{w_j}$ for $j\in \{1, \ldots, n\} \setminus\{i\}$. Note that $\tau_i$ is independent from $u_i$, and the probability that $i$ is included in our set of $k$ samples $\mathcal{S}$ is exactly equal to $\Pr[\frac{u_i}{w_i}\leq \tau_i] = \min(1, \tau_iw_i)$. 
Moreover, conditioned on the event that $i\in S$, we have that  $\tau = \tau_i$. Accordingly, for $i\in S$, $\frac{w_i}{\min(1, \tau w_i)} = \frac{w_i}{\min(1, \tau_i w_i)}$, 
and thus $\hat{w}_i$ can equivalently be written as: 
\vspace{-.3em}
\begin{align*}
     \hat{w}_{i} = \left\{\begin{matrix}\frac{w_i}{\min(1, \tau_i w_i)} \quad\quad\quad\quad i\in S
 \\0 \hfill i\notin S
\end{matrix}\right.
\end{align*}
 \vspace{-1em}
 
\noindent With this definition in place, we prove some intermediate claims. 
\begin{claim}
\label{variance_bound_item}
    $\Var\left[\hat{w}_i\right] \leq w_i\cdot\E\left[\frac{1}{\tau_i}\right]$
\end{claim}
\begin{proof}
We begin by analyzing $\E\left[\hat{w}_i^2\right]$. 
    Conditioning on $\tau_i$, we have:
    \begin{align*}
        \E\left[\hat{w}_i^2 \mid \tau_i\right] = \frac{w_i^2}{\min(1, \tau_i w_i)^2} \cdot \Pr[i \in S]=\frac{w_i^2}{\min(1, \tau_i w_i)^2} \cdot \min(1, \tau_i w_i) = \frac{w_i^2}{\min(1, \tau_i w_i)} = w_i^2 \cdot \max\left(1, \frac{1}{\tau_i w_i}\right).
    \end{align*}
    From the law of total expectation, we thus have that $\E\left[\hat{w}_i^2\right] = w_i^2 \cdot \E\left[\max\left(1, \frac{1}{\tau_i w_i}\right)\right]$.
    Combined with the fact that $\E\left[\hat{w}_i\right]^2 = w_i^2$ (from \Cref{fact:expcov}) we have that $\Var\left[\hat{w}_i\right] = \E\left[\hat{w}_i^2\right] - \E[w_i]^2$ equals:
    \begin{align*}
        \Var\left[\hat{w}_i\right] &= w_i^2 \cdot \left(\E\left[\max\left(1, \frac{1}{\tau_i w_i}\right)\right] - 1\right) = w_i^2 \cdot \E\left[\max\left(0, \frac{1}{\tau_i w_i} -1\right)\right] = w_i^2 \cdot \E\left[\max\left(0, \frac{1}{\tau_i w_i}\right)\right].
    \end{align*}
    And since $\tau_i\cdot w_i$ is non-negative, we obtain:
    \begin{align*}
        \Var\left[\hat{w}_i\right] \leq w_i^2 \cdot \E\left[\frac{1}{\tau_i w_i}\right] = w_i \cdot \E\left[\frac{1}{\tau_i}\right].&\qedhere
    \end{align*}
\end{proof}
\begin{claim}
\label{one_over_tau_expectation}
    $\E\left[\frac{1}{\mathcal{\tau}}\right] \leq \frac{W}{k}$.
\end{claim}
\begin{proof}
Consider the random variable $\hat{W} = \sum_{i=1}^{n} \hat{w}_{i}$. 
% We know from \Cref{fact:expcov} that $\E[\hat{W}] = W$. 
Note that $\hat{W}$ can be rewritten as:
\begin{align*}
\hat{W} = \sum_{i\in S} \hat{w}_{i} = \sum_{i\in S} \frac{w_i}{\min\left(1,\tau w_i\right)} = \sum_{i\in S} \max\left(w_i, \frac{1}{\tau}\right). 
\end{align*}
% So we have that
Hence, 
$\hat{W} \geq \sum_{i\in S} \frac{1}{\tau} = \frac{k}{\tau}$.
We also know from
\Cref{fact:expcov} that $\E[\hat{W}] = W$.
% We also know that $\E[\hat{W}] = W$ (\Cref{fact:expcov}).
If the random variable $\hat{W}$ is always larger than the random variable $\frac{k}{\tau}$, it holds that:
$
\E\left[\frac{k}{\tau}\right] \leq \E\left[\hat{W}\right] = W.$
Dividing by $k$ proves the result.
    % Let's consider the distribution of the uniform variable designated to each item $u=(u_1,u_2,\dots,u_n)$, denoted as $u=(u_1,u_2,\cdots,u_n)\sim D_u$. We can derive $\tau_u$ from this set, where $\tau_u$ represents the (k+1)$^{th}$ smallest value in the set $\frac{u_i}{w_i}$. Furthermore, let's designate a sample set $S_u$.
    % As before, for a given realization of uniform random variables $u_1, \ldots, u_n$, let $\tau$ be the corresponding realization of $\mathcal{T}$ and let $S$ denote the corresponding set of indices sampled. 
    % For any $i$ within set $S$ and assuming $w_i$ takes on a positive value, we can compute $\hat{w}_i$ as follows:
    % \begin{align*}
    %     \hat{w}_i = \frac{w_i}{\Pr[\frac{u_i}{w_i} < \tau_u]} = \frac{w_i}{\min(1,\tau_u\cdot w_i)} = w_i \cdot \max(1,\frac{1}{\tau_u\cdot w_i}) = \max(w_i,\frac{1}{\tau_u}) 
    % \end{align*}
    % Subsequently, the estimator adds all the $\hat{w}_i$ for $i$ within set $S$ and produces $\hat{W}_u$ as the estimator corresponding to the distribution of $u\sim D_u$.
    % \begin{align*}
    %     \hat{W}_u = \sum_{i\in S}\hat{w}_i = \sum_{i\in S} \max(w_i,\frac{1}{\tau_u}) \geq \sum_{i\in S} \frac{1}{\tau_u} \geq \frac{k}{\tau_u} 
    % \end{align*}
    % We can thus observe that for any distribution $u\sim D_u$, we have $\frac{1}{\tau_u} \leq \frac{\hat{W}_u}{k}$. Based on this, we can infer that $\E\left[\frac{1}{\tau}\right] \leq \E\left[\frac{\hat{W}}{k}\right]$. Finally, using \Cref{unbiased_estimator}, we can deduce $\E[\frac{1}{\tau}] \leq \frac{W}{k}$.
\end{proof}

\begin{claim}
\label{one_over_tau_i_expectation}
$\E\left[\frac{1}{\tau_i}\right] \leq \frac{W}{k-1}$
\end{claim}
\begin{proof}
Simply apply \Cref{one_over_tau_expectation} to the setting where we collect $k-1$ priority samples from the set of weights $\{w_j : j\in \{1,\ldots, n\}\setminus \{i\}\}$. Note that $W'=\sum_{j=1,j\neq i}^{n}w_j$ is no larger than $W$, so $\frac{W'}{k-1} \leq \frac{W}{k-1}$.
    % Assume you remove $w_i$ from the set of $w$ and form a new set, $w'=w-\{w_i\}$. It is known that $W'=\sum_{j=1,j\neq i}^{n}w_j$ and set $k'=k-1$. When we apply \Cref{one_over_tau_expectation} to this new set of $w'$ with $k'$, we find that $\tau'$ in this new set corresponds to $\tau_i$. Consequently, we arrive at
    % \begin{align*}
    %     \E\left[\frac{1}{\tau_i}\right] = \E\left[\frac{1}{\tau'}\right] \leq \frac{W'}{k'}= \frac{W'}{k-1} \leq \frac{W}{k-1}
    % \end{align*}
\end{proof}
We are now ready to prove the result of \cite{Szegedy:2006}.
\begin{proof}[Proof of \Cref{thm:mainvar}]
    Applying \Cref{variance_bound_item} and \Cref{one_over_tau_i_expectation}, we have that:
    \begin{align*}
    \sum_{i=1}^{n}\Var\left[\hat{w}_i\right] \leq \sum_{i=1}^{n} w_i\cdot\E\left[\frac{1}{\tau_i}\right] \leq \sum_{i=1}^{n} w_i \frac{W}{k-1} = \frac{W^2}{k-1}.&\qedhere
    \end{align*}
\end{proof}

\section{Discussion and Pedagogical Perspective}
It is natural to ask how the proof above avoids the complexity of \cite{Szegedy:2006}. In fact, it is even simpler than the proof that establishes a looser $O(W^2/k)$ bound from \cite{AlonDuffieldLund:2005}, which invokes a bucketing argument combined with Chernoff bounds. Where's the magic? We do not have a fully satisfying answer, except to point out that a key step in our proof is to reduce the problem to bounding $\mathbb{E}[1/\tau]$. At first glance, this does not seem productive: as the $k^\text{th}$ smallest value of $n$ \emph{scaled} uniform random variables, $\tau$ is a complicated random variable. In particular, its distribution depends on each of $w_1, \ldots, w_n$ in an involved way. However, as we show in \Cref{one_over_tau_expectation}, a simple comparison argument can be used to upper bound $\mathbb{E}[1/\tau]$ without even writing down the probability density function (PDF) of $\tau$. 

This analysis might be interesting from a pedagogical perspective even when all weights are uniform. In this case, $\tau$ is the $(k+1)^\text{st}$ smallest out of $n$ uniform draws. Such random variables appear frequently in course material on randomized algorithms, for example in analyzing the elegant distinct elements algorithm from \cite{Bar-YossefJayramKumar:2002} or when studying the $k$-minimum values (KMV) sketch. In these applications, it is necessary to compute the expected value and variance of $1/\tau$, which typically involves an explicit expression for the PDF of $\tau$ (which is beta distributed), combined with involved calculations \cite{BeyerHaasReinwaldw:2007}. Our approach, on the other hand, gives a simple argument from first principles, which we outline below.

\begin{corollary}\label{cor:unif}
    Let $\tau$ be the $(k+1)^\text{st}$ smallest out of $n$ uniform random variables $u_1, \ldots, u_n$ on $[0,1]$. 
    \begin{align*}
        \E\left[\frac{1}{\tau}\right] &= \frac{k}{n} & &\text{and} & \Var\left[\frac{1}{\tau}\right] &= \frac{n^2 - nk}{k^2(k-1)}.
    \end{align*}
\end{corollary}
%
% \as{Here it was not immediately clear if the argument is about the \cite{Bar-YossefJayramKumar:2002} or \cite{BeyerHaasReinwaldw:2007}. Maybe because I'm not familiar with the first paper (\cite{Bar-YossefJayramKumar:2002}). It is also not immediately clear why we are interested in $\tau_1$. Maybe a few sentences describing the algorithm at a high level would help.} \chris{It's a bit difficult to describe in a concise way... do you think it's okay if we just more explicitly say that in those applications we need to analyze the expected value and variance of $\tau$?} \as{Yeah, stating what we are going to do next helps}
\begin{proof}
Let $\mathcal{S}$ denote the set of $k$ indices $i$ for which $u_i < \tau$. Additionally, let $\tau_1$ equal the $k^\text{th}$ smallest out of $\{u_1, \ldots, u_n\}\setminus\{u_1\}$. There is nothing special about the choice of $1$: $\tau_1$ is simply an auxilary variable used in our analysis. We could have instead chosen the $k^\text{th}$ smallest out of $\{u_1, \ldots, u_n\}\setminus\{u_i\}$ for any $i$. Consider the random variable $X$ defined equivalently (using the same argument as in the previous section) as:
$$
X = \begin{cases}
\frac{1}{\tau} & \text{ if } 1\in \mathcal{S} \\
0 & \text{ otherwise}\end{cases} =\begin{cases}\frac{1}{\tau_1} & \text{ if } 1\in \mathcal{S} \\0 & \text{ otherwise.}\end{cases}
$$
From the second definition, we observe that: 
% \as{Maybe remind why $\Pr[1\in \mathcal{S}\mid\tau_1 = t]$ equals $t$ (needed below)?} \as{I just noticed that in the next page, there is a note about a more general case: $\Pr[i \in \mathcal{S} \mid \tau_i = t] = \min(1, t w_i)$} \chris{I think maybe we can leave out the explanation. I want to get everything to fit on one page :)}
$$
\mathbb{E}[X] 
= \E\left[ \mathbb{E}\left[X\mid \tau_1\right] \right]
= \E\left[ \Pr[1\in \mathcal{S}\mid\tau_1]\cdot \frac{1}{\tau_1}\right]
= 1.
$$
Alternatively, consider the first definition. The value of $\tau$ is independent from the event that $1\in \mathcal{S}$, and by symmetry, $\Pr[1\in \mathcal{S}] = \frac{k}{n}$. So,
$$
\mathbb{E}[X] = \mathbb{E}\left[\frac{1}{\tau}\right]\cdot \Pr[1\in \mathcal{S}] = \mathbb{E}\left[\frac{1}{\tau}\right]\cdot \frac{k}{n}.
$$
We conclude that in order for $\mathbb{E}\left[\frac{1}{\tau}\right]\cdot \frac{k}{n}$ to equal $1$, it must be that
\begin{align}\label{eq:unif_exp}
    \mathbb{E}\left[\frac{1}{\tau}\right] = \frac{n}{k}.
\end{align}
We can then compute the variance of ${1}/{\tau}$ using a similar argument (and applying \eqref{eq:unif_exp} in the last step):
\begin{align*}
\mathbb{E}[X^2] = \E\left[ \mathbb{E}[X^2\mid \tau_1]\right] = \E\left[ \Pr[1\in \mathcal{S}\mid\tau_1]\cdot \frac{1}{\tau_1^2}\right] 
&=E\left[\tau_1\cdot \frac{1}{\tau_1^2}\right] = \mathbb{E}\left[\frac{1}{\tau_1}\right] = \frac{n-1}{k-1}.
\end{align*}
Alternatively, again using that $\tau$ is independent from $\mathbbm{1}\left[1\in \mathcal{S}\right]$, we have that $\mathbb{E}[X^2] = \mathbb{E}\left[\frac{1}{\tau^2}\right] \cdot \frac{k}{n}$
So, we conclude that $\mathbb{E}\left[\frac{1}{\tau^2}\right] = \frac{n(n-1)}{k(k-1)}$. Finally, we have the bound:
\begin{align*}
    \Var\left[\frac{1}{\tau}\right] = \mathbb{E}\left[\frac{1}{\tau^2}\right] - \mathbb{E}\left[\frac{1}{\tau}\right]^2 
    % = \frac{n(n-1)k}{k^2(k-1)} - \frac{n^2(k-1)}{k^2(k-1)} 
    = \frac{n^2 - nk}{k^2(k-1)}.&\qedhere
\end{align*}
\end{proof}
This matches the formula given e.g., in \cite{BeyerHaasReinwaldw:2007}, and establishes that $\Var\left[\frac{1}{\tau}\right] \leq \epsilon^2 \mathbb{E}\left[\frac{1}{\tau}\right]^2$ when $k = O(1/\epsilon^2)$, which is the bound needed to  prove that the \cite{Bar-YossefJayramKumar:2002} distinct elements method gives a $(1\pm \epsilon)$ relative error approximation with $k = O(1/\epsilon^2)$ space 

\section*{Acknowledgements} This work was supported by NSF Award CCF-204623. We thank Xiaoou Cheng for carefully reviewing our proof, and also the anonymous reviewers, whose comments helped improve the presentation of this paper. 

\appendix 
\section{Proof of \texorpdfstring{\Cref{fact:expcov}}{Fact expcov}}
\label{sec:app}
For completeness, we prove the following important and well-known fact about priority sampling, following the approach of existing proofs \cite{DuffieldLundThorup:2007}.
\expcov*
\begin{proof}
Let $\tau_1, \ldots, \tau_n$ be as defined in Section \ref{sec:main}. Recall that $\hat{w}_i$ can equivalently be written as: 
\begin{align*}
     \hat{w}_{i} = \left\{\begin{matrix}\frac{w_i}{\min(1, \tau_i w_i)} \quad\quad\quad\quad i\in S
 \\0 \hfill i\notin S
\end{matrix}\right.
\end{align*}
Then, we can compute its expectation:
\begin{align*}
    \E[w_i] &= \E\left[ \E\left[w_i \mid \tau_i\right] \right]
    = \E\left[\frac{w_i}{\min(1, \tau_i w_i)}\Pr[i \in \mathcal{S} \mid \tau_i] \right] = w_i.
\end{align*}
In the last step we use that $\Pr[i \in \mathcal{S} \mid \tau_i] = \min(1, \tau_i w_i)$, which follows from noting that, for $i$ to be in $\mathcal{S}$, it must be that $\frac{u_i}{w_i}$ is less than $\tau_i$.

Next we show that $\E\left[\hat{w}_i \hat{w}_j\right] = w_i w_j = \E\left[\hat{w}_i\right]\E\left[\hat{w}_j\right]$. 
Let $\tau_{i,j}$ denote the $(k-1)^\text{st}$ smallest value of $\frac{u_r}{w_r}$ for $r\in \{1,\ldots, n\} \setminus \{i,j\}$. 
    % Note that, conditioned on the event that $i,j \in S$, we always have that  $\tau = \tau_{i,j}$. 
    If either $i$ or $j$ is absent from the $S$, then either $\hat{w}_i$ or $\hat{w}_j$ is $0$. So we have:
    \begin{align*}
     \E[\hat{w}_i \hat{w}_j \mid \tau_{i,j}] = \frac{w_i}{\min(1, \tau_{i,j} w_i)} \frac{w_j}{\min(1, \tau_{i,j} w_j)} \Pr[i,j\in S \mid \tau_{i,j}]
    \end{align*}
    We can thus compute the overall expectation as:    
    \begin{align*}
    \E[\hat{w}_i \hat{w}_j] = \E\left[\E\left[\hat{w}_i \hat{w}_j \mid \tau_{i,j}\right]\right]
    &= \E\left[\frac{w_i}{\min(1, \tau_{i,j} w_i)}\cdot\frac{w_j}{\min(1, \tau_{i,j} w_j)}\Pr[i,j \in \mathcal{S} \mid \tau_{i,j}\right] \\
    &= \E\left[\frac{w_i}{\min(1, \tau_{i,j} w_i)}\cdot\frac{w_j}{\min(1, \tau_{i,j} w_j)}\min(1, \tau_{i,j} w_i) \cdot\min(1, \tau_{i,j} w_j)\right] 
    = w_i w_j.
\end{align*}
Above we have used the fact that, for $i$ and $j$ to both be in $s$, it must be that both $\frac{u_i}{w_i}$ and $\frac{u_j}{w_j}$ are less than $\tau_{i,j}$, which happens with probability $\min(1, \tau_{i,j} w_i)\cdot \min(1, \tau_{i,j} w_j)$.
\end{proof}

\section{Comparison to Szegedy's Result and a Refinement}
\label{app:refinement_thrm}
The goal of our work is to prove that the variance of Priority Sampling (which takes exactly $k$ samples) matches that of Threshold Sampling (which takes $k$ samples in expectation) up to a multiplicative factor of just $\frac{k}{k-1}$. In particular, our \Cref{thm:mainvar} proves a variance bound of $\frac{W}{k-1}$ for Priority Sampling, while we show a bound of $\frac{W}{k}$ at the top of page 2 for Threshold Sampling. However, this later bound is not quite tight for two reasons. In our short analysis of Threshold Sampling  to obtain the $\frac{W}{k}$ bound we:
\begin{enumerate}
    \item upper bound $(1-p_i)$ by $1$ for all $p_i$, and
    \item upper bound $\sum_{i \in \mathcal{K}}  w_i \leq \sum_{i \in 1}^n w_i = W$, where we recall that $\mathcal{K}$ contains all $i$ for which $w_i\tau \leq 1$.
\end{enumerate} 
While these steps only lead to minor differences when all $p_i$'s are small, they mean that Threshold Sampling can actually outperform the $W^2/k$ bound when some probabilities are close to $1$. 

A natural question, raised by \cite{DuffieldLundThorup:2007}, is if Priority Sampling still nearly matches Threshold Sampling in such cases. Szegedy answers this question affirmatively: his Theorem 2 proves that Priority Sampling matches Threshold Sampling up to the $\frac{k}{k-1}$ factor for any set of weights. This strengthens our \Cref{thm:mainvar} (his Theorem 4) and, since
Threshold Sampling is known to be optimal among a natural class of sampling methods for any set of weights (see \cite{DuffieldLundThorup:2007}), Szegedy establishes that Priority Sampling is essentially optimal as well.\footnote{The published version of \cite{Szegedy:2006} contains a typographical error suggesting that his Theorem 4 (equivalent to our \cref{thm:mainvar}) implies his Theorem 2; however, this is incorrect as Theorem 2 is in fact stronger and implies Theorem 4.}

While our simple proof approach does not recover Szegedy's full optimality result, we do show that it can obtain a somewhat tighter bound than \Cref{thm:mainvar}. In particular, we obtain the following bound: 
\begin{claim}
\label{clm:tighter}
Let $w_1 \geq \ldots, \geq w_n$ be weights and let $\hat{w}_1, \ldots, \hat{w}_n$ be as defined in \eqref{eq:priority_sampling}, let $\hat{W} = \sum_{i=1}^n \hat{w}_i$, and let ${W} = \sum_{i=1}^n {w}_i$.
 \begin{align*}
 \Var[\hat{W}] = \sum_{i=1}^{n} \Var\left[\hat{w}_i\right] \leq \frac{W^2 - \sum_{i=1}^{n}w_i^2}{k-1} - \sum_{i=k}^{n} w_i^2.
 \end{align*}
\end{claim}
Note that above, and throughout the remainder of this section, we assume without loss of generality that the weights are sorted by magnitude. I.e., that $w_1 \geq \ldots, \geq w_n$.
\Cref{clm:tighter} was stated as Theorem 7 in \cite{Szegedy:2006}, which emphasizes that the result is ``still rather tight'' in comparison to the full optimality result. Importantly, the bound captures that fact that Priority Sampling performs better when there are a few large weights, in which case $\sum_{i=1}^{n}w_i^2$ can be a signifcant fraction of $W^2 = (\sum_{i=1}^{n}w_i)^2$.

We first prove some preliminary results towards establishing \Cref{clm:tighter}.

\begin{claim}
\label{claim:tau_i_inverse_w_i}
Let $\tau_i$ be as defined in \Cref{sec:main}. For every $i \geq k$, $w_i \leq \frac{1}{\tau_i}$.
\end{claim}
\begin{proof}
Recall that $\tau_i$ was defined as the $k^{\text{th}}$ smallest value of $\frac{u_j}{w_j}$ among all $j \neq i$. Recalling that weights are sorted so that $w_1 \geq \ldots, \geq w_n$, since each $u_j \leq 1$, there must be at least $k$ values of $\frac{u_j}{w_j}$ that are less than or equal to $\frac{1}{w_k}$. As such, $\tau_i \leq \frac{1}{w_k} \leq \frac{1}{w_i}$, which proves the claim.
\end{proof}

\begin{claim}
\label{claim:tight_variance_w_hat}
For every $i \geq k$, $\Var\left[\hat{w}_i\right] \leq w_i\cdot\E\left[\frac{1}{\tau_i}\right] - w_i^2$.
\end{claim}
\begin{proof}
We have that:
\begin{align*}
\Var\left[\hat{w}_i\right] &= w_i^2 \cdot \left(\E\left[\max\left(1, \frac{1}{\tau_i w_i}\right)\right] - 1\right) = w_i^2 \cdot \E\left[\max\left(0, \frac{1}{\tau_i w_i} -1\right)\right].
\end{align*}
By \Cref{claim:tau_i_inverse_w_i},  we have that $w_i \leq \frac{1}{\tau_i}$. So $\frac{1}{\tau_i w_i} -1$ is always positive for all $k \leq i$. Thus:
\begin{align*}
\Var\left[\hat{w}_i\right] &= w_i^2 \cdot \E\left[\max\left(0, \frac{1}{\tau_i w_i} -1\right)\right] = w_i^2 \cdot \E\left[\frac{1}{\tau_i w_i} -1\right] = w_i \cdot \E\left[\frac{1}{\tau_i}\right] - w_i^2. \qedhere 
\end{align*}
\end{proof}

\begin{claim}
\label{claim:tight_one_over_tau_i_expectation}
$\E\left[\frac{1}{\tau_i} \right] \leq \frac{W-w_i}{k-1}$
\end{claim}
\begin{proof}
This can be derived direction from our proof of \Cref{one_over_tau_i_expectation}.
\end{proof}
Finally, we are ready to prove our tighter upper bound on the variance of Priority Sampling.
\begin{proof}[Proof of \Cref{clm:tighter}]
Based on \Cref{claim:tight_variance_w_hat} and \Cref{claim:tight_one_over_tau_i_expectation}, we have
\begin{align*}
\text{For every } i \text{ where } i \geq k, \quad &\Var\left[\hat{w}_i\right] \leq w_i \cdot \frac{W - w_i}{k - 1} - w_i^2. \\
\text{For every } i \text{ where } i < k, \quad &\Var\left[\hat{w}_i\right] \leq w_i \cdot \frac{W - w_i}{k - 1}.
\end{align*}
So we conclude that:
\begin{align*}
\sum_{i=1}^{n} \Var\left[\hat{w}_i\right] = \sum_{i=1}^{k-1} \Var\left[\hat{w}_i\right] + \sum_{i=k}^{n} \Var\left[\hat{w}_i\right] &\leq  \sum_{i=1}^{n} w_i \cdot \frac{W-w_i}{k-1} - \sum_{i=k}^{n} w_i^2 \\
&= \sum_{i=1}^{n} \frac{w_i\cdot W-w_i^2}{k-1} - \sum_{i=k}^{n} w_i^2 \\
&= \frac{W^2}{k-1} - \sum_{i=1}^{n} \frac{w_i^2}{k-1} - \sum_{i=k}^{n} w_i^2. \qedhere
\end{align*}
\end{proof}

\bibliographystyle{apalike}
\bibliography{paper}
\end{document}